\documentclass{article}
\usepackage{amsmath, amssymb, amsthm}
\usepackage{mathtools}
\usepackage{graphicx}
\usepackage{booktabs}
\usepackage{enumitem}
\usepackage{float}
\usepackage{microtype}
\usepackage{listings}
\usepackage{natbib}
\usepackage[hidelinks]{hyperref}
\usepackage[nameinlink,noabbrev]{cleveref}
\usepackage[toc,page]{appendix}
\theoremstyle{plain}
\newtheorem{proposition}{Proposition}[section]
\newtheorem{lemma}[proposition]{Lemma}
\newtheorem{assumption}[proposition]{Assumption}
\newtheorem{problem}[proposition]{Problem}
\newtheorem{theorem}[proposition]{Theorem}
\theoremstyle{remark}
\newtheorem*{remark}{Remark}
\numberwithin{equation}{section}

\newcommand{\E}{\mathbb{E}}

\newcommand{\PP}{\mathbb{P}}
\newcommand{\dd}{\mathrm{d}}
\newcommand{\barF}{\bar F}
\newcommand{\haz}{h}
\newcommand{\R}{\mathbb{R}}
\newcommand{\Cost}{\mathcal{C}} 
\setlist[itemize]{leftmargin=*, itemsep=2pt, topsep=2pt}
\setlist[enumerate]{leftmargin=*, itemsep=2pt, topsep=2pt}

\title{Two-Instrument Screening under Soft Budget Constraints}
\author{Xinli Guo \thanks{Email: \texttt{xinlig@mun.ca}. Department of Economics, Memorial University of Newfoundland.}}
\date{August 2025}

\begin{document}

\maketitle

\begin{abstract}
We study soft budget constraints in multi-tier public finance when an upper-tier government uses two instruments: an ex-ante grant schedule and an ex-post rescue. Under convex rescue costs and standard primitives, the three-stage leader--follower problem collapses to one dimensional screening with a single allocation index: the cap on realized rescue. A hazard-based characterization delivers a unified rule that nests (i) no rescue, (ii) a threshold--cap with commitment, and (iii) a threshold--linear--cap without commitment. The knife-edge for eliminating bailouts compares the marginal cost at the origin to the supremum of a virtual weight, and the comparative statics show how greater curvature tightens caps while discretion shifts transfers toward front loading by lowering the effective grant weight. The framework provides a portable benchmark for mechanism design and yields testable implications for policy and empirical work on intergovernmental finance.
\end{abstract}


\medskip
\newpage

\section{Introduction}\label{sec:intro}
Soft budget constraints (SBC) emerge when an upper tier acquires a track record of rescues and lower-tier entities rationally adjust behavior in anticipation of future relief. While the literature spans transition economies and federations, two gaps persist. First, \emph{timing}: transfers are often decided sequentially rather than in a simultaneous-move environment. Second, \emph{instrumentation}: practice mixes an ex-ante grant menu with ex-post rescues, yet most formal benchmarks rely on a single instrument under full commitment.

This paper models SBC as a \emph{two-instrument} Stackelberg screening problem. A leader commits at $t=0$ to a grant schedule and a cap on realized rescues; followers privately observe fiscal-need types, choose effort and fiscal variables, and face noisy gap signals at $t=2$. The realized payout is implementable and capped. 

\paragraph{Contributions.} We provide four general results.
\begin{enumerate}[leftmargin=*, itemsep=2pt]
\item \textbf{Reduction.} The three-stage problem collapses to one-dimensional screening with a monotone allocation index (Proposition~\ref{prop:reduction}).
\item \textbf{General characterization.} For any convex increasing rescue cost $\Cost(x)$, the optimal cap solves the pointwise condition
\[
\Cost'\bigl(b^{\ast}(\theta)\bigr)=\bigl(\gamma\,\omega_b(\theta)/\lambda_T(\theta)\bigr)\,\haz(\theta),
\]
projected to $[0,\bar b]$ and ironed when the virtual term is nonmonotone (see Theorem~\ref{thm:char}). The classic quadratic case is a corollary (Proposition~\ref{prop:opt_cap}).
\item \textbf{A unified knife-edge.} A self-consistent no-rescue regime obtains iff
\[
\Cost'(0^+)\;\ge\;\sup_{\theta}\frac{\gamma\,\omega_b(\theta)}{\lambda_T(\theta)}\,\haz(\theta),
\]
which nests the linear-cost threshold used in applied work (Proposition~\ref{prop:nobailout}).
\item \textbf{Commitment vs.\ discretion.} Without commitment at $t{=}2$ and with a convex continuation loss from residual gaps, realized rescues are \emph{threshold--linear--cap}; the interior slope lowers the effective grant weight $\lambda_T$ proportionally to the probability of being on the interior branch, strengthening SBC incentives (Section~\ref{subsec:disc}).
\end{enumerate}

\paragraph{Roadmap.}
Section~\ref{sec:lit} reviews related work.
Section~\ref{sec:model} introduces the three-stage environment and shows a reduction to one-dimensional screening with a cap as the allocation index.
Section~\ref{sec:stage_game} delivers the hazard-based characterization; Theorem~\ref{thm:char} is the main result.
Section~\ref{sec:Tstar} develops the optimal transfer schedule and closed forms; key results include Proposition~\ref{prop:opt_cap} (quadratic benchmark) and Proposition~\ref{prop:nobailout} (knife-edge for no rescue).
Section~\ref{subsec:disc} contrasts commitment with discretion and derives the threshold--linear--cap rule.
Appendices collect proofs and computational details.

\section{Literature Review}\label{sec:lit}

The paper intersects three strands of work:  
(i)~soft budget constraints (SBC) in multi-tier public finance,  
(ii)~mechanism design with Stackelberg leadership, and  
(iii)~municipal-finance empirics.

\subsection{Soft Budget Constraints}

The SBC idea begins with \citet{kornai1986}.  Early formalizations
\citep{kornai_maskin_roland_2003} show that ex post efficient bailouts
undermine ex ante effort and borrowing; see also
\citet{dewatripont_maskin_1995} for a dynamic commitment model.
Applications include
\citet{weingast_1995},
\citet{bordignon_2001}, and the survey
by \citet{goodspeed_2016}.
Recent papers ask \emph{when} an upper tier can credibly refuse rescues:
\citet{AmadorEtAl2021AER} derive fiscal limits under limited commitment,
\citet{PavanSegal2023JET} study repeated screening, and
\citet{AcemogluJackson2024RESTUD} analyze relational contracts with hidden
actions.  Yet these models stop short of a closed-form, implementable transfer
rule under generic convex cost. We fill that gap by providing a hazard-based
characterization with an implementable payout convention.

\subsection{Mechanism Design and Stackelberg Leadership}

Incentive-compatible grant design dates back to
\citet{bordignon_montolio_picconi_2003} under simultaneous moves.
\citet{toma_2013} introduces leader--follower timing with full information.
\citet{chen_silverman_2019} obtain threshold payments in a one-shot model with asymmetric costs, yet ignore ex post instruments and dynamic credibility.  
Our contribution is a \textit{two-instrument} Stackelberg screen that nests bailout and
no bailout regimes in a single hazard-based condition and allows for general convex rescue costs.
\paragraph{Placement within mechanism design.}
Our approach follows the virtual-surplus tradition of \citet{Myerson1981} and the general toolkit in \citet{LaffontMartimort2002}.
Envelope and differentiation arguments rely on \citet{MilgromSegal2002}.
Relative to models that emphasize limited commitment or relational enforcement---such as \citet{AmadorEtAl2021AER}, \citet{PavanSegal2023JET}, and \citet{AcemogluJackson2024RESTUD}---our contribution is an \emph{implementable two-instrument screening benchmark} that (i) yields a closed-form, hazard-based cap rule under generic convex rescue costs; (ii) provides an explicit knife-edge for no rescue; and (iii) nests commitment and discretion via a unified single-index representation.

\subsection{Municipal Finance Empirics}

Empirical work examines fiscal capacity and service costs \citep{bird_2012,sancton_govcfs_2014}, tax-base sharing
\citep{dahlby_ferede_2021}, and borrowing limits
\citep{found_tompson_2020}. Evidence on municipal-level SBC is growing, e.g.\ \citet{BraccoDoyle2024JPubE}, and cross-country evidence in \citet{Rodden2006}. Our theory offers a benchmark that links bailout policy to mechanism design and clarifies the micro-data needed for identification.

\section{Model}\label{sec:model}

We study the interaction between a single upper-tier government $P$ and a
continuum of local jurisdictions $i\in\mathcal I\subset[0,1]$. 

\paragraph{Type convention.}
Throughout, the private type $\theta\in[\underline\theta,\bar\theta]$ is a \emph{fiscal-need index}:
a higher $\theta$ corresponds to a weaker local tax base / higher per-unit service cost,
hence a larger underlying funding gap.

\subsection{Technologies and Preferences}

\begin{description}[
  style=nextline,
  leftmargin=0pt,
  labelsep=0pt,
  itemindent=0pt,
  itemsep=0pt, parsep=0pt, topsep=2pt,
  font=\bfseries 
]
\item\textbf{Basic services} Each jurisdiction \( i \) delivers a bundle of essential public services
\( q \). The monetary cost is modeled by \( C(q,\theta) \), where \( \theta \)
denotes fiscal need, with higher \( \theta \) implying higher marginal cost.
\item\textbf{Heterogeneous fiscal need} \( \theta \) is drawn
from a continuous distribution \( f(\theta) \) with support
\( [\underline{\theta},\bar{\theta}] \).
\item\textbf{Local effort} Effort \( e \ge 0 \) generates own-source revenue \( R(e,\theta) \) with $R'_e>0$ and $R''_{ee}<0$.
\item\textbf{Disutility of effort} Effort imposes a linear utility cost 
\begin{align*}
\text{Disutility}=-\phi e, \qquad \phi>0.
\end{align*}
\item\textbf{Service and investment utility}
$B(q)$ captures household utility from $q$, while $\Gamma(I)$ is the longer-run payoff from capital $I$, both twice differentiable with diminishing returns.
\end{description}

\paragraph{Transfers}
The leader’s three instruments are:
\begin{enumerate}[label=(\roman*)]
\item \textbf{Unconditional grant} $\tau$ (maps to $T$), set \emph{ex ante};
\item \textbf{Matching transfer} at share $s$ on capital outlays $I$;
\item \textbf{Ex-post bailout} $b\ge 0$ if a realized gap remains after
fiscal shocks $\varepsilon$ are realized.
\end{enumerate}
Before local choices, the leader commits to $(\tau,s,\bar D,\beta,b)$, where $\beta(\cdot)$ is a signal-based payout rule at $t=2$ and $b(\cdot)$ is a type-based cap.

Given these transfers, the \textbf{one-period cash-flow constraint} at $t=1$ (before any payout) is
\begin{align}
  G &= \bigl[C\bigl(q,\theta\bigr)+(1-s)I+rD\bigr] 
        -\bigl[R(e,\theta)+\tau+g+sI+D\bigr] + \varepsilon . \label{eq:gap_pre}
\end{align}
\footnote{Notation reminder: in the reduced-form mechanism we write $T(\theta)$ for the unconditional operating grant that corresponds to the empirical instrument $\tau$; and $g$ captures largely exogenous capital transfers that are treated as constants in screening.}

If $G>0$ a funding gap exists.  The leader observes a noisy signal $\hat{G}=G+\eta$ and pays the \emph{realized payout}
\[
  p(\hat G,\hat\theta)
  \;=\; \mathbf 1\{\hat G>0\}\cdot\min\{\beta(\hat G),b(\hat\theta),\hat G\}
  \;\in\;[0,\hat G],
\]
at $t=2$ under commitment to $(\beta,b)$.

\subsection{Local effort $e$: incentives and marginal condition}\label{sub:effort}

\begin{assumption}[Primitives]\label{ass:primitives}
Types $\theta$ lie in $[\underline{\theta},\bar{\theta}]$ with density $f>0$.
The local cost and revenue functions satisfy, for all $\theta$,
\[
C'_q(q,\theta)>0,\qquad C''_{qq}(q,\theta)\ge0,\qquad
R'_e(e,\theta)>0,\qquad R''_{ee}(e,\theta)<0.
\]
\end{assumption}

\begin{assumption}[Observation \& rule regularity]\label{ass:regularity}
The signal rule $\beta$ is nondecreasing and a.e.\ differentiable with slope in $[0,1)$; the audit noise $\eta$ has a continuous density $f_\eta$ with bounded tails; and $(e,\theta)\mapsto R'_e(e,\theta)$ is continuous. 
\end{assumption}

\begin{assumption}[Signal MLRP]\label{ass:mlrp}
For $\theta_2>\theta_1$, the family $\{\hat G\mid\theta\}$ satisfies MLRP, so for any nondecreasing $\varphi$, $\E[\varphi(\hat G)\mid\theta_2]\ge \E[\varphi(\hat G)\mid\theta_1]$. 
\end{assumption}

\begin{assumption}[IFR on types]\label{ass:IFR}
The type distribution has increasing failure rate: $\,\haz(\theta)=f(\theta)/\barF(\theta)\,$ is weakly increasing on $[\underline\theta,\bar\theta]$. This assumption is used to ensure monotonic implementability (and ironing) of the allocation; the pointwise characterization in Theorem~\ref{thm:char} does not require IFR.
\end{assumption}

\begin{assumption}[Restriction to threshold signal rules]\label{ass:beta-threshold}
We restrict attention to nondecreasing, piecewise-constant (threshold) signal-based payout rules $\beta$. This class is consistent with administrative practice and eliminates effort--report interactions almost everywhere.
\end{assumption}

\paragraph{Observation frictions and default.}
Default occurs iff $p(\hat G,\hat\theta)<G$. Because
of information frictions, municipalities rationally expect a positive rescue probability $\pi>0$. Anticipating the chance of a bailout, they optimally reduce tax effort $e$ and rely more on debt $D$.\footnote{For axiomatic treatments of decision under noisy or imperfect perception, see \citet{PivatoVergopoulos2020JME}.}

\begin{assumption}[Effort independence via threshold $\beta$]\label{ass:report-invariance}
Under Assumptions~\ref{ass:primitives}--\ref{ass:beta-threshold}, the first-order condition for effort on the cap-slack branch contains no term depending on the report $\hat\theta$; $e^{\ast}(\theta)$ is report-independent up to boundary-density terms on the cap-binding set.
\end{assumption}

\noindent\emph{Technical detail.} See Lemma~\ref{lem:e-indep} in Appendix~\ref{app:proofs}.

\subsection{Annual timeline: three decision stages}\label{subsec:timeline}

We normalize one fiscal year to $t\in\{0,1,2\}$; the sequence repeats every year.

\paragraph{Stage $t=0$ (policy commitment).}
The leader announces $\Pi=(\tau,s,\bar D,\beta,b)$ where $\beta:\R_+\to\R_+$ is the signal-based payout rule implemented at $t=2$ and $b(\cdot)$ is a type-based cap.

\paragraph{Stage $t=1$ (local choices and gap realization).}
After observing its private type $\theta\sim f$, the municipality selects effort $e\ge0$, capital $I\ge0$, and debt $0\le D\le\bar D$. A mean-zero fiscal shock $\varepsilon$ is realized and the pre-payout gap is $G$ from \eqref{eq:gap_pre}.

\paragraph{Stage $t=2$ (signal, payout, default test).}
The leader observes $\hat G=G+\eta$ and pays $p(\hat G,\hat\theta)=\mathbf 1\{\hat G>0\}\min\{\beta(\hat G),b(\hat\theta),\hat G\}$. Default occurs iff $p(\hat G,\hat\theta)<G$.

\subsection{Local optimization at $t=1$}

Given policy $\Pi$, the municipality solves
\begin{equation}\label{eq:local_prog}
\begin{aligned}
\max_{e,q,I,D}\; & 
  \E_{\varepsilon,\eta}\!\Bigl[
      B(q)+\Gamma(I)+R(e,\theta)-\phi\,e
      -\varphi\,\mathbf1\!\{\,p(\hat G,\hat\theta)<G\,\}
      +\omega_b\,p(\hat G,\hat\theta)
  \Bigr] \\[-2pt]
\text{s.t.}\quad
&G = C(q,\theta)+(1-s)I+rD
      -\bigl[R(e,\theta)+\tau+g+sI+D\bigr]
      +\varepsilon, \\[4pt]
&0 \le D \le \bar D,\qquad e,q,I \ge 0 .
\end{aligned}
\end{equation}

\noindent
Since $\partial \E[p(\hat G,\hat\theta)]/\partial e=-R'_e(e,\theta)\,\E[\beta'(\hat G)\,\mathbf1\{\beta(\hat G)<b(\hat\theta)\}]$, the interior FOC on the signal branch is
\begin{equation}\label{eq:FOC_e_correct}
  R'_e\!\bigl(e^{\ast}(\theta);\theta\bigr)\,
  \Bigl\{
    1
    +\varphi\,\E\!\left[f_\eta\!\bigl(\hat{G}-\beta(\hat{G})\bigr)\bigl(1-\beta'(\hat G)\bigr)\right]
    -\omega_b\,\E\!\bigl[\beta'(\hat G)\,\mathbf1\{\beta(\hat G)<b(\hat\theta)\}\bigr]
  \Bigr\}
  = \phi .
\end{equation}
For threshold (piecewise constant) $\beta$, $\beta'(\hat G)=0$ a.e., hence the last term vanishes.


\subsection{Reduction to One--Dimensional Screening}
\label{subsec:reduction}
\begin{remark}[On rare cap binding]\label{rem:cap-slack-optional}
Our main differentiation steps rely only on continuous audit noise and threshold (piecewise-constant) $\beta$, which make boundary sets Lebesgue-null. A stronger “rare cap binding” condition can be imposed as a robustness convenience but is not required for the results; see Lemma~\ref{lem:boundary-bound}.
\end{remark}
\paragraph{Step~1. Local optimization.}
Fix $(\tau,s,\bar D,\beta)$ and a true type $\theta$. Minimizing the pre-bailout resource block delivers $(q^{\ast},I^{\ast},D^{\ast})$ and reduced cost $C_0(\theta)$, independent of the report.

\paragraph{Step~2. Effort choice.}
Effort $e^\ast(\theta)$ is pinned down by \eqref{eq:FOC_e_correct} (cap slack almost everywhere).

\paragraph{Step~3. Quasi-linear reduced form with payout cap.}
Define the \emph{expected payout under cap}, conditional on type,
\[
  \tilde b(\hat\theta;\theta)
   = \E\bigl[\min\{\beta(\hat G),b(\hat\theta)\}\,\big|\,\theta\bigr].
\]
Then the interim utility from reporting $\hat\theta$ can be written
\begin{equation}
  U_L(\hat\theta,\theta)=\lambda_T(\theta)\,T(\hat\theta)
    + \omega_b\,\tilde b(\hat\theta;\theta)+K(\theta),
  \label{eq:UL-reduced}
\end{equation}
with a grant weight
\[
  \lambda_T(\theta)\;=\;\omega_T\; -\; \omega_b\,\E\Bigl[\frac{\partial p}{\partial T}\,\Big|\,\theta\Bigr].
\]
Under Assumption~\ref{ass:beta-threshold} and continuous noise, $\beta'(\hat G)=0$ almost everywhere, so $\lambda_T(\theta)=\omega_T$.

\begin{lemma}[Grant crowd-out factor]\label{lem:lambdaT}
With $\hat G=G+\eta$ and $G$ decreasing in $T$ one-for-one, the marginal effect of $T$ on the realized payout is
\[
  \frac{\partial}{\partial T}\E[p(\hat G,\hat\theta)\mid\theta]
  = -\,\E\!\left[\beta'(\hat G)\,\mathbf 1\{\beta(\hat G)<b(\hat\theta)\}\,\middle|\,\theta\right]
  + \text{boundary terms}.
\]
If $\beta$ is threshold (piecewise constant), $\beta'(\hat G)=0$ a.e. and the boundary terms vanish under continuous noise, hence $\lambda_T(\theta)=\omega_T$.
If $\beta$ has an interior linear branch with slope $m\in(0,1)$ on the cap-slack set, then $\lambda_T(\theta)=\omega_T-\omega_b\,m\cdot \PP_\theta[\text{cap slack and } \hat G \text{ in linear range}]$.
\end{lemma}

\begin{proposition}[Reduction]\label{prop:reduction}
Under Assumptions~\ref{ass:primitives}, \ref{ass:regularity}, \ref{ass:mlrp} and~\ref{ass:beta-threshold}, the original
moral-hazard problem is equivalent to a direct mechanism in which
municipal interim utility is quasi-linear in the \emph{cap parameter} $b$ through $\tilde b(\hat\theta;\theta)=\E[\min\{\beta(\hat G),b(\hat\theta)\}\mid\theta]$ as in \eqref{eq:UL-reduced}.
\end{proposition}

\begin{remark}[Constant marginal utility of bailouts]\label{rem:omega-const}
Lemma~\ref{lem:lambdaT} implies that, under threshold $\beta$, the grant weight equals $\omega_T$. If instead a discretionary linear segment applies (Section~\ref{subsec:disc}), the weight falls below $\omega_T$ proportionally to the slope and the probability of being on that linear branch.
\end{remark}

\begin{remark}[Effect of default-loss term and boundary control]\label{rem:defaultloss}
Because the default indicator flips only on the cap-binding set, boundary contributions are controlled by Lemma~\ref{lem:boundary-bound} and are uniformly bounded by a constant times the cap-binding tail probability $\PP_\theta[\beta(\hat G)\ge b(\theta)]$. Hence the marginal effect of a report on the $-\varphi\,\mathbf1\{p<G\}$ term is negligible whenever this tail probability is small, without imposing a separate “rare cap binding” assumption.
\end{remark}

\subsection{Single--Period Mechanism Design}
\label{sec:stage_game}
We henceforth work with the reduced form \eqref{eq:UL-reduced}, and with the leader's cost taken in \emph{expectation} over the realized payout.
\paragraph{Preferences and objectives.}
The local government's interim utility is quasi-linear with marginal weights $\omega_T(\theta)$ on grants and $\omega_b(\theta)$ on realized bailouts. These weights enter the analysis only through the IC/envelope terms.
\emph{Notation.} To avoid confusion with the local production cost $C(q,\theta)$, we denote the Province's \emph{rescue resource cost} by $\Cost(x)$.
\paragraph{Leader's cost.}
Given a cap profile $b(\theta)$, the Province's expected resource cost at type $\theta$ is
\[
  \E\Bigl[ \Cost\bigl(\min\{\beta(\hat G),b(\theta)\}\bigr) \Bigm|\,\theta\Bigr]
  \;+\; \gamma\,T(\theta),
\]
with $\Cost$ convex, increasing and $\Cost'(0^+)=\alpha$. The preference weights $\omega_T,\omega_b$ appear only in the local government's IC via $\lambda_T(\theta)$ and do not enter the Province's resource-cost objective directly.

\paragraph{Envelope and weights.}
With $V(\theta)=U_L(\theta,\theta)$ and $V(\underline\theta)=\underline U$,
\begin{equation}\label{eq:envelope-correct}
V'(\theta)=\lambda_T'(\theta)\,T(\theta)+\omega_b\,\partial_\theta\tilde b(\theta;\theta)+K'(\theta).
\end{equation}

\begin{remark}[Baseline: effective grant weight on the local side]\label{rem:lambdaT-const}
Under threshold $\beta$ with continuous audit noise, the \emph{effective grant weight in the local government's IC/envelope} equals $\lambda_T(\theta)\equiv\omega_T$; it is not a weight in the Province's resource-cost objective. Hence \eqref{eq:envelope-correct} reduces to
\[
V'(\theta)=\omega_b\,\partial_\theta\tilde b(\theta;\theta)+K'(\theta).
\]
If the $t{=}2$ rule has an interior linear branch with slope $m\in(0,1)$ on the cap-slack set, then
$\lambda_T(\theta)=\omega_T-\omega_b\,m\cdot \PP_\theta[\text{cap slack and }\hat G\text{ in the linear range}]$.
\end{remark}

\begin{assumption}[Single crossing in the allocation index]\label{ass:sc}
Define the allocation index for report $\hat\theta$ at true type $\theta$ by
\[
  x(\hat\theta;\theta)\;=\;\lambda_T(\theta)\,T(\hat\theta)\;+\;\omega_b(\theta)\,\tilde b(\hat\theta;\theta),
\]
so that interim utility is $U_L(\hat\theta,\theta)=x(\hat\theta;\theta)+K(\theta)$ with $K$ absolutely continuous.
Assume the single-crossing condition in $x$:
\[
  \frac{\partial^2 U_L}{\partial\theta\,\partial x}(\hat\theta,\theta)\;\ge\;0\quad\text{for all }(\hat\theta,\theta).
\]
Under Assumption~\ref{ass:mlrp}, this holds because $\partial_\theta \tilde b(\hat\theta;\theta)\ge 0$ for nondecreasing $\beta$.
\end{assumption}

\begin{problem}[Leader’s program --- reduced form]\label{prob:leader}
\[
  \min_{T(\cdot),\,b(\cdot)}
    \; \E_\theta\!\Bigl[\,\E\bigl[\Cost\bigl(\min\{\beta(\hat G),b(\theta)\}\bigr)\bigm|\theta\bigr]
    + \gamma\,T(\theta)\Bigr]
\]
\[
\text{s.t.}\quad
\begin{cases}
  \text{(IC)} & V'(\theta)=\lambda_T'(\theta)\,T(\theta)+\omega_b\,\partial_\theta\tilde b(\theta;\theta)+K'(\theta),\\[0.5ex]
  \text{(IR)} & V(\theta)\ge\underline U,\quad\forall\theta, \\[0.5ex]
  \text{(LL)} & 0\le T(\theta),\; 0\le b(\theta)\le\bar b,\quad\forall\theta .
\end{cases}
\]
\end{problem}

\begin{theorem}[Characterization under convex rescue cost]\label{thm:char}
Suppose Assumptions~\ref{ass:primitives}--\ref{ass:sc} hold. Then there exists an IC--IR--LL optimal mechanism with a nondecreasing cap schedule $b^\ast(\theta)$ such that, at almost every $\theta$,
\[
\Cost'\bigl(b^\ast(\theta)\bigr)
\;=\;
\frac{\gamma\,\omega_b(\theta)}{\lambda_T(\theta)}\,\haz(\theta),
\qquad \haz(\theta)=\frac{f(\theta)}{\barF(\theta)},
\]
with projection to $[0,\bar b]$ and ironing where the virtual term is nonmonotone. Under threshold $\beta$ with continuous audit noise, $\lambda_T(\theta)\equiv\omega_T$.
\end{theorem}

\noindent\textit{Corollaries.} The main characterization yields two direct corollaries:  
(i) the quadratic case (Proposition~\ref{prop:opt_cap}); and  
(ii) the no-rescue knife-edge (Proposition~\ref{prop:nobailout}).

\section{Optimal Transfer Schedule}\label{sec:Tstar}

\begin{lemma}[Conditional cap--min calculus]\label{lem:capcalc}
Let $F_\beta(\cdot\mid\theta)$ be the c.d.f.\ of $\beta(\hat G)$ conditional on type with continuous density. For any cap $b\ge0$,
\[
  \frac{\partial}{\partial b}\E\bigl[\min\{\beta(\hat G),b\}\mid\theta\bigr]
  \;=\; \PP_\theta\!\bigl[\beta(\hat G)\ge b\bigr],\qquad
  \frac{\partial}{\partial b}\E\bigl[\min\{\beta(\hat G),b\}^{2}\mid\theta\bigr]
  \;=\; 2b\,\PP_\theta\!\bigl[\beta(\hat G)\ge b\bigr].
\]
\end{lemma}

\begin{proposition}[Quadratic case; closed-form]\label{prop:opt_cap}
If $\Cost(x)=\alpha x+\tfrac{\kappa}{2}x^2$ and $\lambda_T\equiv \omega_T$, then
\[
b^{\ast}(\theta)=\min\!\Bigl\{\bar b,\max\!\bigl\{0,\; \kappa^{-1}\!\bigl((\gamma\omega_b/\omega_T)\,\haz(\theta)-\alpha\bigr)\bigr\}\Bigr\},
\]
and $T^{\ast}$ is determined by the index differential (see \eqref{eq:T-star}). This recovers the threshold--cap geometry and the triple-zone rule with cutoffs $\theta^{\min}$ and $\theta^{\dagger}$ defined in~\eqref{eq:theta-dagger}.
\end{proposition}

\begin{equation}\label{eq:b-star}
b^{\ast}(\theta)=\min\!\Bigl\{\bar b,\max\!\bigl\{0,\; \kappa^{-1}\bigl((\gamma\omega_b/\lambda_T)\,\haz(\theta)-\alpha\bigr)\bigr\}\Bigr\}.
\end{equation}

\begin{equation}\label{eq:T-star}
\dd T(\theta)=-(\omega_b/\omega_T)\,\dd\tilde b(\theta),\qquad T(\theta^{\min})=0.
\end{equation}

\begin{equation}
\theta^{\min} \;=\; \inf\{\theta:\, b^{\ast}(\theta)>0\}, 
\qquad
\theta^{\dagger} \;=\; \inf\{\theta:\, b^{\ast}(\theta)=\bar b\}.
\label{eq:theta-dagger}
\end{equation}

\begin{proposition}[General no-rescue knife-edge]\label{prop:nobailout}
Under the conditions of Theorem~\ref{thm:char}, a self-consistent no-rescue regime $b^{\ast}(\theta)\equiv 0$ is optimal iff
\[
\boxed{\, \Cost'(0^+)\;\ge\;\sup_{\theta}\frac{\gamma\,\omega_b(\theta)}{\lambda_T(\theta)}\,\haz(\theta)\, }.
\]
Equivalently, if $\sup_{\theta}\frac{\gamma\,\omega_b(\theta)}{\lambda_T(\theta)}\,\haz(\theta)\le \Cost'(0^+)$ then $b^\ast\equiv 0$; otherwise $b^\ast>0$ on a set of positive measure.
\end{proposition}

\paragraph{IR normalization and LL implications.}
Normalize $V(\theta^{\min})=\underline U$ and note $\tilde b^{\ast}(\theta^{\min})=0$, whence $T^{\ast}(\theta^{\min})=0$. Because of the negative relation \eqref{eq:T-star}, the LL requirement $T\ge0$ implies that whenever $\tilde b^{\ast}(\theta)>0$ on some region, the optimal $T^{\ast}(\theta)$ is driven to the boundary $T=0$ there, shifting screening to $b(\cdot)$.

\begin{proposition}[Second–best efficiency]\label{prop:welfare}
Under Assumptions \ref{ass:primitives}–\ref{ass:sc}, the allocation in Theorem~\ref{thm:char} is second–best efficient among IC–IR–LL mechanisms; the proof follows a virtual-surplus argument in Appendix~\ref{app:proofs}.
\end{proposition}

\subsection{Comparative statics}\label{subsec:CS2}

Let $\haz(\theta)=f(\theta)/\bar F(\theta)$ and $\lambda_T=\omega_T$ in the baseline. The interior zero solves $\haz(\theta^{\min})=\alpha\,\lambda_T/(\gamma\,\omega_b)$. By the implicit function theorem,
\begin{align*}
  \frac{\partial\theta^{\min}}{\partial\alpha}
      &=\frac{1}{\haz'(\theta^{\min})}\,\frac{\lambda_T}{\gamma\,\omega_b}>0,\qquad
  \frac{\partial\theta^{\min}}{\partial\omega_b}
      =\frac{1}{\haz'(\theta^{\min})}\,\Bigl(-\frac{\alpha\,\lambda_T}{\gamma\,\omega_b^{2}}\Bigr)<0,\\[6pt]
  \frac{\partial\theta^{\min}}{\partial\lambda_T}
      &=\frac{1}{\haz'(\theta^{\min})}\,\frac{\alpha}{\gamma\,\omega_b}>0,\qquad
  \frac{\partial\theta^{\min}}{\partial\gamma}
      =\frac{1}{\haz'(\theta^{\min})}\,\Bigl(-\frac{\alpha\,\lambda_T}{\gamma^{2}\,\omega_b}\Bigr)<0.
\end{align*}

For the interior cap $b_{\max}=(\gamma\omega_b/\lambda_T-\alpha)/\kappa$,
\[
  \frac{\partial b_{\max}}{\partial\kappa}
      =-\,\frac{1}{\kappa^{2}}
         \Bigl(\frac{\gamma\omega_b}{\lambda_T}-\alpha\Bigr)<0,\quad
  \frac{\partial b_{\max}}{\partial\lambda_T}
      =-\,\frac{\gamma\omega_b}{\kappa\lambda_T^{2}}<0,\quad
  \frac{\partial b_{\max}}{\partial\gamma}
      =\frac{\omega_b}{\kappa\lambda_T}>0.
\]

\begin{remark}[Interpretation]\hfill\break
\begin{enumerate}[label=\textbf{\arabic*.}, leftmargin=1.8em, itemsep=2pt, topsep=2pt]
\item \textbf{Trigger boundary.} Increasing $\alpha$ or $\lambda_T$ raises $\theta_{\min}$ (harder to trigger); increasing $\omega_b$ or $\gamma$ lowers $\theta_{\min}$ (easier to trigger).
\item \textbf{Interior cap.} Under the quadratic cost $\mathcal C(x)=\alpha x+\tfrac{\kappa}{2}x^2$, larger $\kappa$ or $\lambda_T$ implies smaller $b_{\max}$; larger $\omega_b$ or $\gamma$ implies larger $b_{\max}$.
\item \textbf{Discretion.} If discretion lowers the effective weight, e.g.,
$\lambda_T^{\mathrm{disc}}=\omega_T-\omega_b\,m\cdot\Pr[\text{interior}]$ with $m\in(0,1)$,
then by the chain rule $\frac{\partial \theta_{\min}}{\partial m}<0$ and $\frac{\partial b_{\max}}{\partial m}>0$.
Thus discretion expands the set of types receiving a positive bailout and raises the interior cap (weaker institutional discipline).
\end{enumerate}
\end{remark}

\section{Policy Implications}\label{sec:policy}

\subsection{Design principles}

\begin{enumerate}[label=\textbf{P\arabic*},leftmargin=*]

\item \textbf{Codify a \emph{triple-zone} rule.}  
      Proposition~\ref{prop:opt_cap} together with \eqref{eq:theta-dagger}
implies a simple menu in the quadratic baseline: \textit{(i)} no transfer when the reported type is below
      $\theta^{\min}$;  
      \textit{(ii)} a flat-to-rising cap on $[\theta^{\min},\,\theta^{\dagger}]$. When the $t{=}2$ rule has an interior slope (discretion), the effective grant weight falls below $\omega_T$, reinforcing the condition.
\item \textbf{A single inequality decides whether bailouts survive.}  
      Under the no-rescue candidate, Proposition~\ref{prop:nobailout} shows that bailouts disappear when
      \[
      \Cost'(0^+) \;\ge\; \sup_{\theta}\frac{\gamma\,\omega_b(\theta)}{\lambda_T(\theta)}\,\haz(\theta).
      \]
\item \textbf{Front-load under softness (discretion).}  
      When the $t{=}2$ payout rule has an interior slope (Section~\ref{subsec:disc}), the effective grant weight falls below $\omega_T$ proportionally to the slope and the probability of being on that interior branch (Lemma~\ref{lem:lambdaT}).
\item \textbf{Make the cap bite by increasing curvature.}  
      Higher curvature of $\Cost$ around the origin (e.g.\ larger $\kappa$ in the quadratic case) reduces $b^\ast$ and tightens the cap.
\end{enumerate}

\subsection{Limited-liability regions}
From \eqref{eq:T-star}, $T^{\ast}(\theta)$ is (weakly) decreasing in $\tilde b^{\ast}(\theta)$.
Hence on any type region where $\tilde b^{\ast}(\theta)>0$, the grant LL constraint ($T\ge0$) typically binds, pushing $T^{\ast}(\theta)$ to $0$ and shifting screening to $b(\cdot)$.
Consequently, interior $T^{\ast}>0$ arises only (i) on the no-rescue region where $b^{\ast}=0$ (i.e.\ $\theta<\theta^{\min}$), or (ii) on ironed segments of the virtual weight when ironing is required under IFR.

\section{Conclusion}\label{sec:conclusion}
This paper recasts upper--lower tier rescues as a two-instrument screening problem with an implementable payout convention
\[
p(\hat G,\hat\theta)=\mathbf{1}\{\hat G>0\}\min\{\beta(\hat G),\,b(\hat\theta),\,\hat G\}.
\]
Four takeaways emerge:
\begin{enumerate}[leftmargin=*]
\item \textbf{One-dimensional reduction.} Under convexity and MLRP, the multi-stage, moral-hazard environment reduces to one-dimensional adverse selection with a \emph{cap parameter} as the allocation index.
\item \textbf{General cap rule.} With any convex rescue cost $\Cost$, the IC--IR--LL optimum is characterized by $\Cost'\bigl(b^\ast(\theta)\bigr)=(\gamma\omega_b/\lambda_T)\haz(\theta)$, with cutoffs pinned down by the hazard $\haz(\theta)$; the quadratic case yields a closed-form \emph{threshold--cap} (Proposition~\ref{prop:opt_cap}).
\item \textbf{Unified regime test.} A self-consistent no-rescue regime obtains iff $\Cost'(0^+) \ge \sup_\theta \frac{\gamma\,\omega_b(\theta)}{\lambda_T(\theta)}\,\haz(\theta)$ (Proposition~\ref{prop:nobailout}).
\item \textbf{Discretion vs.\ commitment.} Without commitment at $t{=}2$, the realized rule becomes \emph{threshold--linear--cap}; the interior slope lowers the effective grant weight and strengthens SBC incentives (Section~\ref{subsec:disc} and Lemma~\ref{lem:lambdaT}).
\end{enumerate}

\section*{Acknowledgements}
I wish to thank Professor Marcus Pivato and Professor Hélène Huber for their valuable suggestions and encouragement. I also acknowledge the helpful discussions with peers at Université Paris 1 Panthéon-Sorbonne, which contributed to improving the clarity of the analysis. All remaining errors are my own.

\newpage
\setcounter{table}{0}
\begin{appendices}

\section{Monte Carlo check with endogenous $\lambda_T$}\label{sec:mc}

We complement the closed-form illustration with a Monte Carlo check that endogenizes the discretion weight $\lambda_T$ via the fixed point
\[
\lambda_T \;=\; \omega_T \;-\; \omega_b\, m \cdot \Pr\!\big(0<b^{\ast}(\theta)<\bar b\big),
\]
where $m\in(0,1)$ is the slope of the interior segment in the $t{=}2$ rule. Given $\lambda_T$, the optimal cap under quadratic rescue cost is
\[
b^{\ast}(\theta)=\Big[\, \tfrac{2}{\kappa}\,\tfrac{\gamma\omega_b}{\lambda_T}\,\theta - \tfrac{\alpha}{\kappa}\,\Big]_0^{\bar b},
\]
so that $\theta^{\min}=\alpha/\big(2(\gamma\omega_b/\lambda_T)\big)$ and
$\theta^{\dagger}=(\kappa\bar b+\alpha)/\big(2(\gamma\omega_b/\lambda_T)\big)$.

\paragraph{Design.}
We draw $N=200{,}000$ types from a Weibull$(k{=}2,\text{scale}{=}1)$ prior (with hazard $\haz(\theta)=2\theta$), set
$(\omega_T,\omega_b,\gamma,\alpha,\kappa,\bar b)=(1,0.8,1,0.2,1,0.8)$ and $m=0.5$,
and iterate on $\lambda_T$ until convergence. For each iteration we compute
$p_{\text{int}}=\Pr(0<b^{\ast}(\theta)<\bar b)$ and update
$\lambda_T\leftarrow \omega_T-\omega_b m\,p_{\text{int}}$.

\paragraph{Findings.}
The fixed point yields $\lambda_T^{\text{disc}}\approx 0.897$ and
$p_{\text{int}}\approx 0.258$. The Monte Carlo estimates of the cutoffs closely match the
closed forms; Fig.~\ref{fig:bstar-mc} overlays the MC binned means on the theoretical $b^{\ast}(\theta)$,
and Table~\ref{tab:numeric_mc} compares theoretical and simulated thresholds.

\begin{table}[H]
\centering
\caption{Monte Carlo vs.\ closed-form thresholds (Weibull-$k{=}2$).}
\label{tab:numeric_mc}
\begin{tabular}{@{}lcccc@{}}
\toprule
Regime & $\lambda_T$ & $\theta^{\min}$ (theory / MC) & $\theta^{\dagger}$ (theory / MC) & $\Pr(0<b^{\ast}<\bar b)$ \\
\midrule
Commitment & $1.000$ & $0.125\ /\ 0.125$ & $0.625\ /\ 0.625$ & $0.308$ \\
Discretion & $0.897$ & $0.112\ /\ 0.112$ & $0.562\ /\ 0.562$ & $0.258$ \\
\bottomrule
\end{tabular}
\end{table}

\begin{figure}[H]
  \centering
  \includegraphics[width=.85\textwidth]{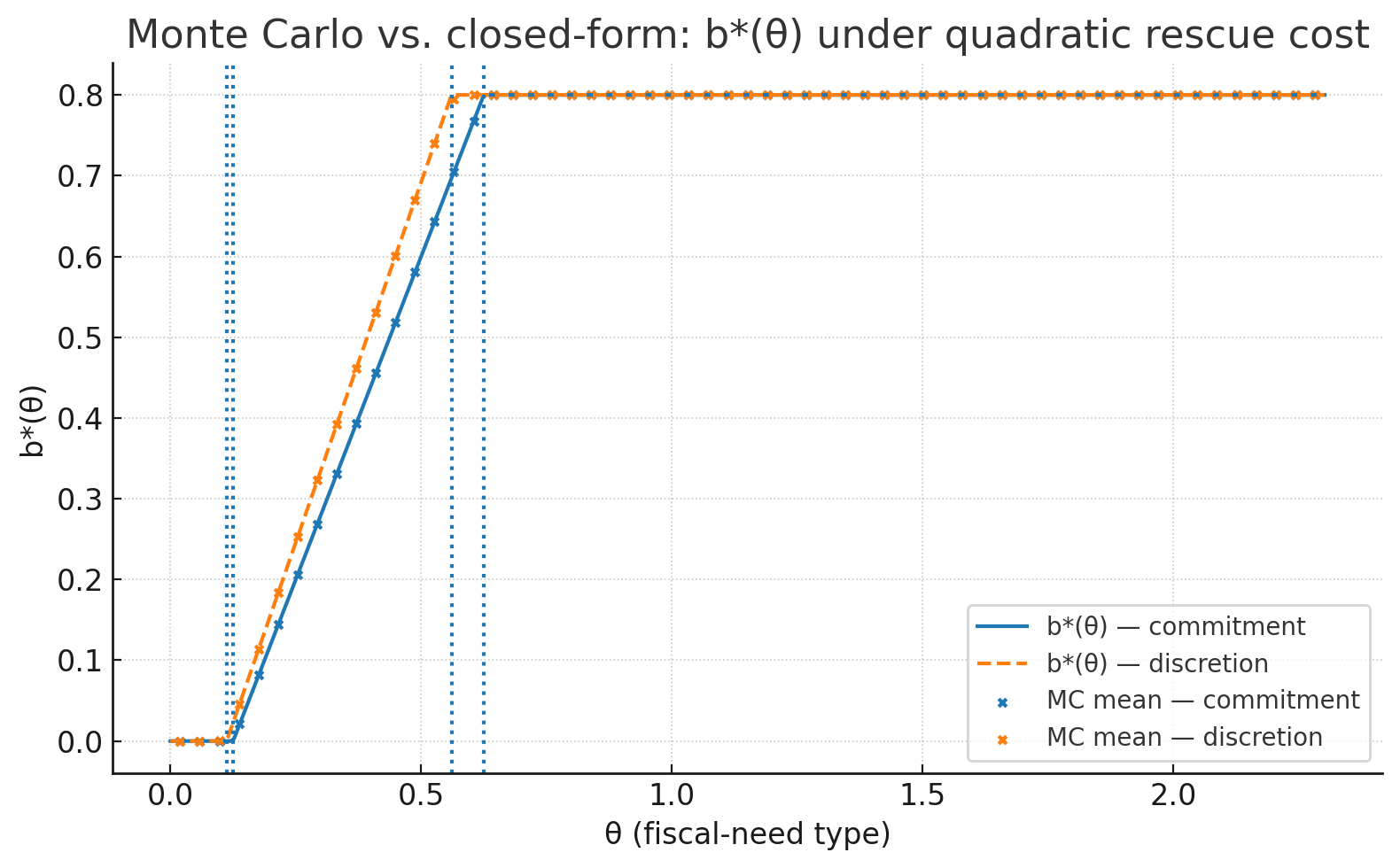}
  \caption{Monte Carlo (binned means) vs.\ closed-form $b^{\ast}(\theta)$ under commitment and discretion.
  Vertical dotted lines mark $\theta^{\min}$ and $\theta^{\dagger}$ in each regime.}
  \label{fig:bstar-mc}
\end{figure}

\paragraph{Remark.}
This experiment does not require simulating audit noise or effort explicitly; it checks implementability geometry through $\Pr(0<b^{\ast}<\bar b)$ and the induced change in $\lambda_T$.
If desired, one can add a continuous audit noise and simulate the implementable payout
$p(\hat G,\hat\theta)=\mathbf 1\{\hat G>0\}\min\{\beta(\hat G),b^{\ast}(\theta),\hat G\}$; the thresholds in $b^{\ast}(\cdot)$ are unchanged.

\section{Symbols used in the model and empirical discussion}\label{tab:symbols}
\begin{table}[H]
\centering
\small
\caption{Symbols used in the model and empirical discussion.}
\label{tab:notation}
\begin{tabular}{@{}lp{9.2cm}@{}}
\toprule
\textbf{Symbol} & \textbf{Description} \\
\midrule
$i$ & Local jurisdiction index. \\
$\theta$ & Fiscal need / gap type (higher $=$ weaker tax base, larger need). \\
$f(\theta),F(\theta)$ & Density and c.d.f.\ of types on $[\underline\theta,\bar\theta]$. \\
$\barF(\theta)$ & Survivor function $1-F(\theta)$. \\
$e$ & Local revenue effort. \\
$R(e,\theta)$ & Own-source revenue function. \\
$q$ & Level of basic services. \\
$C(q,\theta)$ & Cost to produce $q$ given $\theta$. \\
$\tau$ & Unconditional operating grant; maps to $T$ in the model. \\
$g$ & Predictable capital transfer. \\
$s$ & Provincial cost-share rate for capital $I$. \\
$I$ & Local capital investment. \\
$D$ & New debt (subject to approval); $\bar D$ debt limit. \\
$r$ & Debt-service factor on $D$. \\
$G$ & Ex post fiscal gap before payout. \\
$\beta(\hat G)$ & Signal-based payout rule at $t{=}2$. \\
$b(\theta)$ & Type-based cap at $t{=}0$. \\
$p(\hat G,\hat\theta)$ & Realized payout $\mathbf 1\{\hat G>0\}\min\{\beta(\hat G),b(\theta),\hat G\}$. \\
$T(\theta)$ & Ex-ante grant schedule. \\
$\tilde b(\hat\theta;\theta)$ & Expected payout under cap: $\E[\min\{\beta(\hat G),b(\hat\theta)\}\mid\theta]$. \\
$\omega_T,\,\omega_b$ & Marginal utilities of $T$ and realized payout. \\
$\alpha,\kappa$ & Rescue cost parameters: $\Cost(x)=\alpha x+\frac{\kappa}{2}x^{2}$. \\
$\lambda_T(\theta)$ & Effective weight on $T$ in screening: $\omega_T$ (threshold baseline). \\
$U_L$ & Local government utility (interim). \\
$U_P$ & Province’s (negative) expected cost. \\
$V(\theta)$ & Truthful utility $U_L(\theta,\theta)$. \\
$\underline U$ & Reservation utility (IR constraint). \\
$\theta^{\min},\theta^{\dagger}$ & Lower/upper cutoffs for the optimal cap. \\
$b_{\max}$ & Interior bailout cap level. \\
$\pi$ & Ex-ante default-coverage probability (descriptive). \\
$\delta$ & Ex-ante default probability. \\
$\phi$ & Welfare loss to residents under unresolved gap. \\
$\bar b$ & Statutory cap on per-period bailout. \\
$\eta$ & Audit/report noise. \\
$B(q),\;\Gamma(I)$ & Utility/benefit from service level $q$ and investment $I$. \\
$\chi$ & Convex loss parameter in discretionary rescue (Appendix~\ref{subsec:disc}). \\
\bottomrule
\end{tabular}
\end{table}

\section{Discretionary rescue at $t=2$ and backward induction}\label{subsec:disc}

\paragraph{Provincial problem at $t=2$ (no commitment).}
Suppose the Province cannot commit to $\beta$ at $t=0$ and instead chooses 
a payout $x$ at $t=2$ after observing the noisy gap signal $\hat G$. 
For tractability, let the loss from an unresolved residual gap $(\hat G-x)_{+}$ be convex:
\[
  L\bigl((\hat G-x)_{+}\bigr) \;=\; \frac{\chi}{2}\,(\hat G-x)_{+}^{2}, 
  \qquad \chi>0.
\]
The Province solves
\[
  \min_{\,0\le x \le \min\{\bar b,\,\hat G\}}\;
      \alpha x \;+\; \frac{\kappa}{2}x^{2}
      \;+\; L\bigl((\hat G-x)_{+}\bigr).
\]
When $0<x<\min\{\bar b,\,\hat G\}$ the FOC is
$\alpha+\kappa x - \chi(\hat G-x) = 0$, hence
\[
  \beta^{\text{disc}}(\hat G)
   \;=\; \left[\, \frac{\chi\,\hat G - \alpha}{\kappa+\chi} \,\right]_{[\,0,\,\bar b\,]}.
\]
Therefore $\beta^{\text{disc}}$ is \emph{threshold--linear--cap} in $\hat G$.

\paragraph{Backward induction to $t=1$.}
Municipalities at $t=1$ anticipate $\beta^{\text{disc}}(\cdot)$ and choose effort accordingly.
On the interior linear branch where $\beta^{\text{disc}}{}'(\hat G)=\chi/(\kappa+\chi)$, 
the default-probability component in \eqref{eq:FOC_e_correct} is scaled by $\kappa/(\kappa+\chi)$
and there is an additional marginal-rescue term $-\omega_b\,\E[\beta^{\text{disc}}{}'(\hat G)\mathbf1\{\beta^{\text{disc}}<b\}]$.

\paragraph{Discussion.}
This discretionary benchmark microfounds a threshold--linear--cap rule at $t=2$ and shows how the slope filters into 
\eqref{eq:FOC_e_correct}, strengthening the soft budget moral-hazard channel. 
Our commitment baseline avoids time inconsistency by fixing $\beta$ at $t=0$; 
the discretion variant is useful as a robustness check.

\section{Variable marginal utility of bailouts}\label{app:omega-variable}

When the marginal utility of a realized bailout, $\omega_b(\theta)$, varies
across jurisdictions—for instance because political pressure is stronger for
small communities—the first-order condition for the optimal cap becomes
\[
  \omega_b(\theta)\;=\;\alpha\;+\;\kappa\,b^{\ast}(\theta),
\]
so that the linear segment in~\eqref{eq:b-star} reads
\[
  b^{\ast}(\theta)
  \;=\;\frac{\omega_b(\theta)-\alpha}{\kappa},
  \qquad
  0\;\le\;b^{\ast}(\theta)\;\le\;\bar b.
\]
\textit{Implication.} As long as $\omega_b(\theta)$ is weakly increasing in
$\theta$, the cap schedule remains monotone and retains the
\emph{threshold–linear–cap} geometry under the discretionary benchmark.  The slope may now vary with type; for
empirical calibration one needs an estimate of $\omega_b(\theta)$, e.g.\ from
survey weights or past voting patterns.

\section{Single-crossing and monotonicity details}\label{app:sc-proof}
With $U_L(\hat\theta,\theta)=x(\hat\theta;\theta)+K(\theta)$ and $\partial^2 U_L/\partial\theta\,\partial x\ge0$,
the Spence--Mirrlees single-crossing property implies standard IC inequalities:
for any $\theta>\hat\theta$,
\[
  \bigl[U_L(\theta,\theta)-U_L(\hat\theta,\theta)\bigr]
  \;\ge\;
  \bigl[U_L(\theta,\hat\theta)-U_L(\hat\theta,\hat\theta)\bigr].
\]
Since $K$ cancels, this reduces to
$x(\theta;\theta)-x(\hat\theta;\theta)\ge x(\theta;\hat\theta)-x(\hat\theta;\hat\theta)$.
By letting reports be truthful on the RHS, we get $x(\theta)\ge x(\hat\theta)$, hence monotonicity.
When the virtual term $\bigl(\gamma\omega_b/\lambda_T(\theta)\bigr)\haz(\theta)$ fails to be increasing,
standard ironing (\`a la Myerson) delivers a nondecreasing ironed index.

\section{Regularity for differentiation under the expectation}\label{app:regularity-diff}
We justify the steps leading to \eqref{eq:FOC_e_correct} and Lemma~\ref{lem:margprob}.

\begin{lemma}[Differentiation under the expectation with threshold rules]\label{lem:milgrom-segal}
With continuous audit noise and threshold (piecewise $C^1$) $\beta$, boundary sets are Lebesgue-null, so dominated convergence applies and no cap-slack assumption is required for the differentiation steps below.
For any integrand $\varphi(\hat G,e)$ dominated by an integrable envelope and piecewise $C^1$ in $e$, the map $e\mapsto \E[\varphi(\hat G,e)]$ is a.e.\ differentiable and 
\[
\frac{\dd}{\dd e}\,\E[\varphi(\hat G,e)]=\E[\partial_e \varphi(\hat G,e)].
\]
Moreover, for events of the form $\{\hat G-\beta(\hat G)>0\}$, the boundary set has Lebesgue measure zero, so boundary contributions vanish under dominated convergence.
\end{lemma}

\begin{lemma}[Boundary-effect bound]\label{lem:boundary-bound}
Let $\eta$ have a continuous density with integrable tails and let $\beta$ be threshold (piecewise constant). For any integrand $\varphi(\hat G,e)$ dominated by an integrable envelope and any locally bounded $R'_e$, there exists a constant $C<\infty$ (depending only on the envelope and $\sup|R'_e|$ on compact sets) such that
\[
\Bigl|\tfrac{\dd}{\dd e}\E[\varphi(\hat G,e)] \;-\; \tfrac{\dd}{\dd e}\E\bigl[\varphi(\hat G,e)\,\mathbf 1\{\beta(\hat G)<b(\hat\theta)\}\bigr] \Bigr|
\;\le\; C\,\PP_\theta\!\bigl[\beta(\hat G)\ge b(\hat\theta)\bigr].
\]
In particular, the boundary contribution vanishes whenever the cap-binding tail probability is zero and is uniformly dominated by that tail probability otherwise.
\end{lemma}

\paragraph{Dominated convergence / Leibniz rule.}
Assume: (i) $\eta$ has a continuous density $f_\eta$ with bounded tails; (ii) $\beta$ is piecewise $C^1$ with slope in $[0,1)$ and bounded image; (iii) $R'_e(e,\theta)$ is continuous and locally bounded uniformly in $e$ on compact sets. Then, for any integrable function $g(\hat G,e)$ that is piecewise $C^1$ in $e$ and dominated by an integrable envelope, we may differentiate inside the expectation by dominated convergence / Leibniz’s rule:
\[
\frac{\partial}{\partial e}\,\E\bigl[g(\hat G,e)\bigr]
=\E\Bigl[\frac{\partial}{\partial e}g(\hat G,e)\Bigr].
\]

\paragraph{Indicators and boundary sets.}
For events of the form $\{\hat G-\beta(\hat G)>0\}$, the boundary $\{\hat G-\beta(\hat G)=0\}$ has Lebesgue measure zero because $\eta$ has a density and $\beta$ is a.e.\ differentiable with bounded slope; hence the derivative of the indicator contributes no boundary term.
On threshold rules, $\beta'(\hat G)=0$ a.e., so the marginal-rescue term vanishes, yielding the expressions stated in Lemma~\ref{lem:margprob} and \eqref{eq:FOC_e_correct}.

\section{Technical Lemmas and Proofs}\label{app:proofs}

\noindent\textit{Note (implementable payout).} Throughout, the realized payout is
\[
p(\hat G,\hat\theta)=\mathbf 1\{\hat G>0\}\min\{\beta(\hat G),b(\hat\theta),\hat G\}.
\]
On the cap–slack and positive–signal set where $\beta(\hat G)<b(\hat\theta)$, all derivatives below coincide with those under $p=\beta(\hat G)$.
When the $\min$ picks $\hat G$, boundary sets have Lebesgue measure zero under continuous noise, so the derivative contributions vanish a.e.

\begin{lemma}[Effort independence from report]\label{lem:e-indep}
Under Assumptions~\ref{ass:primitives}, \ref{ass:regularity} and \ref{ass:beta-threshold}, the interior first-order condition \eqref{eq:FOC_e_correct} can be rewritten
\[
R'_e\!\bigl(e^{\ast}(\theta),\theta\bigr)\bigl\{1+\varphi\,\Lambda\bigr\}=\phi,
\qquad
\Lambda=\E\bigl[f_\eta(\hat G-\beta(\hat G))\bigr].
\]
All terms on the right depend only on the \emph{true} type~$\theta$; hence
$e^\star(\theta)$ is independent of the reported~$\hat\theta$ up to boundary-density terms on the cap-binding tail.
\end{lemma}

\begin{lemma}[Marginal default probability on the signal branch]\label{lem:margprob}
Under Assumptions~\ref{ass:primitives}--\ref{ass:regularity}, let $\delta=\PP_\theta\!\bigl[p(\hat G,\hat\theta)<G\bigr]$. On the set where $\beta(\hat G)<b(\hat\theta)$ (cap slack),
\[
  \frac{\partial \delta}{\partial e}
  \;=\;
  -\,R'_e(e,\theta)\,\E\!\left[f_\eta\!\bigl(\hat{G}-\beta(\hat{G})\bigr)\,\bigl(1-\beta'(\hat G)\bigr)\right],
\]
and for \emph{threshold} rules (piecewise constant $\beta$), $\beta'(\hat G)=0$ a.e., so
\[
\frac{\partial \delta}{\partial e}=-R'_e(e,\theta)\,\E\!\left[f_\eta\!\bigl(\hat{G}-\beta(\hat{G})\bigr)\right].
\]
\end{lemma}

\begin{proof}[Proof of Lemma~\ref{lem:lambdaT}]
Recall $\hat G=G+\eta$ and $\partial_T \hat G=\partial_T G=-1$. Write
\[
p(\hat G,\hat\theta)=\mathbf 1\{\hat G>0\}\min\{\beta(\hat G),b(\hat\theta),\hat G\}.
\]
Since $\beta(0)=0$ and $\beta'\in[0,1)$, for $\hat G\ge 0$ we have $\beta(\hat G)\le \hat G$. Hence on $\{\hat G>0\}$ and cap–slack $\{\beta(\hat G)<b(\hat\theta)\}$, the minimum is $\beta(\hat G)$ and
\[
\partial_T p = \beta'(\hat G)\,\partial_T \hat G = -\,\beta'(\hat G).
\]
On the cap-binding set $\{\beta(\hat G)\ge b(\hat\theta)\}$, $p=b(\hat\theta)$ so $\partial_T p=0$. Therefore,
\[
\frac{\partial}{\partial T}\E[p(\hat G,\hat\theta)\mid\theta]
= -\,\E\!\left[\beta'(\hat G)\,\mathbf 1\{\beta(\hat G)<b(\hat\theta)\}\,\middle|\,\theta\right]+\text{boundary terms}.
\]
The boundary terms arise only when the identity of the minimizer $\operatorname*{arg\,min}\{\beta(\hat G),\, b(\hat\theta),\, \hat G\}$ changes; 
since $\eta$ has a continuous density and $\beta$ is a.e.\ differentiable with slope $<1$, 
those switch sets have Lebesgue measure zero and their contribution vanishes under dominated convergence. Hence $\lambda_T(\theta)=\omega_T-\omega_b\,\partial_T\E[p|\theta]$ reduces to the stated expressions; in particular, for threshold $\beta$ we obtain $\lambda_T(\theta)\equiv \omega_T$. \qedhere
\end{proof}

\begin{lemma}[Monotonicity of the allocation index]\label{lem:mono-fixed}
Under IC and Assumption~\ref{ass:sc}, the implemented allocation index
\[
  x(\theta)\;=\;\lambda_T(\theta)\,T(\theta)+\omega_b(\theta)\,\\tilde b(\theta;\theta)
\]
is weakly increasing in $\theta$. When ironing is needed (IFR with nonmonotone virtual term), the ironed allocation preserves nondecreasingness.
\end{lemma}

\begin{lemma}[Monotonicity under IFR and caps]\label{lem:mono_cap}
If $\lambda_T$ and $\omega_b$ are locally constant and $f/\barF$ is
increasing (IFR), then the optimal cap $b^{\ast}(\theta)=\min\{\bar b,\max\{0,\tilde b(\theta)\}\}$ is weakly increasing,
where $\tilde b(\theta)=\frac{\gamma\omega_b}{\kappa\lambda_T}\frac{f(\theta)}{\barF(\theta)}-\frac{\alpha}{\kappa}$.
If $\lambda_T(\theta)$ varies, a sufficient condition is that
$\lambda_T(\theta)$ is weakly decreasing and $f(\theta)/\barF(\theta)$ is
increasing; otherwise, apply standard ironing on the virtual term $\frac{\gamma\omega_b}{\lambda_T(\theta)}\frac{f(\theta)}{\barF(\theta)}$.
\end{lemma}

\begin{proof}[Proof of Eq.~\eqref{eq:FOC_e_correct} (first-order condition for $e$)]
Fix $(\tau,s,\bar D,\beta,b)$ and true type $\theta$. The municipality’s $t{=}1$ objective as a function of $e$ (dropping terms independent of $e$) is
\[
\Phi(e)\;=\;\E\!\Big[R(e,\theta)-\phi e\;-\;\varphi\,\mathbf 1\!\{p(\hat G,\hat\theta)<G\}\;+\;\omega_b\,p(\hat G,\hat\theta)\Big],
\]
with $\hat G=G(e)+\eta$ and $G(e)=C(\cdot)-[R(e,\theta)+\tau+g]+\ldots$ so that $\partial_e \hat G=\partial_e G=-R'_e(e,\theta)$.

\smallskip
\noindent\textbf{Step 1 (Justifying differentiation under $\E$).}
By Assumptions~\ref{ass:regularity}--\ref{ass:beta-threshold}, $\eta$ has a continuous density $f_\eta$ with bounded tails, $\beta$ is piecewise $C^1$ with slope in $[0,1)$ and bounded image, and $R'_e$ is continuous and locally bounded. Hence all integrands below admit a uniform integrable envelope, so dominated convergence / Leibniz rule applies and we may interchange $\partial_e$ and $\E$.

\smallskip
\noindent\textbf{Step 2 (Derivative of the default indicator).}
Define $H(e,\eta)=\hat G-\beta(\hat G)$. On the set where $\beta$ is differentiable,
\[
\partial_e H(e,\eta)\;=\;(\partial_e \hat G)\,\bigl(1-\beta'(\hat G)\bigr)\;=\;-R'_e(e,\theta)\,\bigl(1-\beta'(\hat G)\bigr).
\]
Approximate the Heaviside $\mathbf 1\{u>0\}$ by smooth $s_n(u)$ with $s_n'\to \delta_0$ in the sense of distributions, and apply dominated convergence:
\[
\frac{\partial}{\partial e}\E\!\big[\mathbf 1\{H(e,\eta)>0\}\big]
=\lim_{n\to\infty}\E\!\big[s_n'(H)\,\partial_e H\big]
= \E\!\left[\delta_0\!\big(H\big)\,\partial_e H\right].
\]
Since $H=\hat G-\beta(\hat G)$ is a monotone $C^1$ transformation of $\hat G$ with slope $1-\beta'(\hat G)\in(0,1]$ a.e., the density of $H$ at $0$ equals $f_\eta\!\big(\hat G-\beta(\hat G)\big)$ a.e. Hence
\[
\frac{\partial}{\partial e}\E\!\big[\mathbf 1\{p(\hat G,\hat\theta)<G\}\big]
=\frac{\partial}{\partial e}\E\!\big[\mathbf 1\{H>0\}\big]
= -\,R'_e(e,\theta)\,\E\!\Big[f_\eta\!\big(\hat G-\beta(\hat G)\big)\,\bigl(1-\beta'(\hat G)\bigr)\Big],
\]
where we have used that on the cap–slack set $\{\beta(\hat G)<b(\hat\theta)\}$ the event $\{p<G\}$ coincides with $\{H>0\}$, while on the cap-binding set the boundary contribution is controlled by Lemma~\ref{lem:boundary-bound}.

\smallskip
\noindent\textbf{Step 3 (Derivative of the realized payout).}
Write $p(\hat G,\hat\theta)=\mathbf 1\{\hat G>0\}\min\{\beta(\hat G),b(\hat\theta),\hat G\}$. Since $\beta(0)=0$ and $\beta'\in[0,1)$, we have $\beta(\hat G)\le \hat G$ for all $\hat G\ge 0$; thus on $\{\hat G>0\}$ and cap–slack $\{\beta(\hat G)<b(\hat\theta)\}$, $p=\beta(\hat G)$ and
\[
\partial_e p \;=\; \beta'(\hat G)\,\partial_e \hat G \;=\; -\,\beta'(\hat G)\,R'_e(e,\theta).
\]
On the cap-binding set $p=b(\hat\theta)$ so $\partial_e p=0$; hence
\[
\frac{\partial}{\partial e}\E[p(\hat G,\hat\theta)]
= -\,R'_e(e,\theta)\,\E\!\Big[\beta'(\hat G)\,\mathbf 1\{\beta(\hat G)<b(\hat\theta)\}\Big] \;+\; \Delta_{\text{bdry}},
\]
where $|\Delta_{\text{bdry}}|\le C\cdot \PP_\theta[\beta(\hat G)\ge b(\hat\theta)]$ by Lemma~\ref{lem:boundary-bound}.

\smallskip
\noindent\textbf{Step 4 (FOC).}
Collecting terms,
\[
\Phi'(e)=R'_e(e,\theta)-\phi-\varphi\,\frac{\partial}{\partial e}\E[\mathbf 1\{p<G\}] +\omega_b\,\frac{\partial}{\partial e}\E[p].
\]
Using the expressions above and canceling the common factor $R'_e(e,\theta)$, the boundary contribution is controlled by Lemma~\ref{lem:boundary-bound} and the first-order condition $\Phi'(e^\ast)=0$ becomes
\[
R'_e\!\bigl(e^{\ast}(\theta),\theta\bigr)\Big\{1+\varphi\,\E\!\big[f_\eta(\hat G-\beta(\hat G))\bigl(1-\beta'(\hat G)\bigr)\big]-\omega_b\,\E\!\big[\beta'(\hat G)\,\mathbf 1\{\beta(\hat G)<b(\hat\theta)\}\big]\Big\}=\phi,
\]
which is Eq.~\eqref{eq:FOC_e_correct}. \end{proof}

\begin{proof}[Proof of Lemma~\ref{lem:margprob}]
Let $\delta(e)=\PP_\theta[p(\hat G,\hat\theta)<G]$. On the cap–slack event $\{\beta(\hat G)<b(\hat\theta)\}$ we have $\{p<G\}=\{\hat G-\beta(\hat G)>0\}= \{H>0\}$. Repeating the mollifier argument in the proof of Eq.~\eqref{eq:FOC_e_correct},
\[
\delta'(e)=\frac{\partial}{\partial e}\E[\mathbf 1\{H>0\}]=\E\!\left[\delta_0(H)\,\partial_e H\right]
= -\,R'_e(e,\theta)\,\E\!\left[f_\eta\!\big(\hat G-\beta(\hat G)\big)\,\bigl(1-\beta'(\hat G)\bigr)\right].
\]
For \emph{threshold} $\beta$ we have $\beta'(\hat G)=0$ a.e., hence
\[
\delta'(e)=-\,R'_e(e,\theta)\,\E\!\left[f_\eta\!\big(\hat G-\beta(\hat G)\big)\right].
\]
On the cap-binding set, the event $\{p<G\}$ becomes $\{b(\hat\theta)<G\}$ and contributes a boundary term controlled by Lemma~\ref{lem:boundary-bound}, uniformly bounded by a constant times the cap-binding tail probability; this does not alter the formula. \qedhere
\end{proof}

\begin{proof}[Proof of Proposition~\ref{prop:reduction} (implementation equivalence via a cap index)]
Fix $(\tau,s,\bar D,\beta)$ and a true type $\theta$. Minimizing over $(q,I,D)$ yields a reduced cost $C_0(\theta)$ independent of the report. By Assumption~\ref{ass:report-invariance}, the interior $e^\ast(\theta)$ is report-independent, and any boundary contribution is controlled by Lemma~\ref{lem:boundary-bound}. Hence interim utility can be written as
\[
U_L(\hat\theta,\theta)=\lambda_T(\theta)\,T(\hat\theta)+\omega_b\,\tilde b(\hat\theta;\theta)+K(\theta),
\]
with $\lambda_T(\theta)=\omega_T-\omega_b\,\partial_T\E[p(\hat G,\hat\theta)\mid\theta]$ and $\tilde b(\hat\theta;\theta)=\E[\min\{\beta(\hat G),b(\hat\theta)\}\mid\theta]$. By Lemma~\ref{lem:lambdaT}, under threshold $\beta$ and continuous noise we have $\partial_T\E[p\,|\,\theta]=0$, so $\lambda_T(\theta)=\omega_T$ a.e., delivering the quasi-linear reduced form.

Single crossing (Assumption~\ref{ass:sc}) then implies the allocation index
\[
x(\hat\theta;\theta)=\lambda_T(\theta)T(\hat\theta)+\omega_b(\theta)\tilde b(\hat\theta;\theta)
\]
is nondecreasing in $\hat\theta$, with ironing if needed. Conversely, given any nondecreasing cap schedule $b(\cdot)$, revelation/taxation principles with a public signal imply outcome-equivalent implementation by the realized payout
\[
p(\hat G,\hat\theta)=\mathbf 1\{\hat G>0\}\min\{\beta(\hat G),b(\hat\theta),\hat G\}
\]
and an appropriate $T(\cdot)$, completing the equivalence.
\end{proof}

\begin{proof}[Proof of Lemma~\ref{lem:mono-fixed}]
Let $x(\hat\theta;\theta)=\lambda_T(\theta)T(\hat\theta)+\omega_b(\theta)\tilde b(\hat\theta;\theta)$ and $U_L(\hat\theta,\theta)=x(\hat\theta;\theta)+K(\theta)$. By Assumption~\ref{ass:sc}, $\partial^2 U_L/\partial\theta\,\partial x\ge 0$ (Spence–Mirrlees). Suppose, towards a contradiction, that there exist $\theta_2>\theta_1$ with $x(\theta_2)<x(\theta_1)$. Then IC implies
\[
U_L(\theta_2,\theta_2)\ge U_L(\theta_1,\theta_2),\qquad
U_L(\theta_1,\theta_1)\ge U_L(\theta_2,\theta_1).
\]
Subtracting and using single crossing yields $x(\theta_2)\ge x(\theta_1)$, a contradiction. Hence $x(\theta)$ is weakly increasing. When ironing is needed (IFR with nonmonotone virtual term), the ironed allocation preserves nondecreasingness. \qedhere
\end{proof}

\begin{proof}[Proof of Lemma~\ref{lem:mono_cap}]
Fix $\theta$ and consider an incremental increase of $b(\theta)$ by $\dd b$ while holding $b$ elsewhere fixed. By Lemma~\ref{lem:capcalc}, the marginal increase in the Province’s expected cost at type $\theta$ equals
\[
(\alpha+\kappa b(\theta))\,\PP_\theta[\beta(\hat G)\ge b(\theta)]\,\dd b.
\]
By Myerson’s envelope for direct mechanisms, the marginal (virtual) benefit from relaxing the cap at $\theta$ equals
\[
\frac{\gamma}{\lambda_T(\theta)}\,\omega_b(\theta)\,\haz(\theta)\,\PP_\theta[\beta(\hat G)\ge b(\theta)]\,\dd b,
\]
where $\haz(\theta)=f(\theta)/\bar F(\theta)$ under IFR. Equating marginal cost and benefit cancels the common tail probability and yields the pointwise KKT condition
\[
\alpha+\kappa b(\theta)=\frac{\gamma\,\omega_b(\theta)}{\lambda_T(\theta)}\,\haz(\theta).
\]
If $\lambda_T,\omega_b$ are locally constant, this implies $b(\theta)=\frac{1}{\kappa}\Big(\frac{\gamma\omega_b}{\lambda_T}\haz(\theta)-\alpha\Big)$, which is increasing in $\theta$ because $h$ is increasing under IFR. Projection onto $[0,\bar b]$ preserves weak monotonicity. If $\lambda_T(\theta)$ varies with $\theta$, a sufficient condition for the RHS to be weakly increasing is that $\lambda_T$ be weakly decreasing while $h$ is weakly increasing; otherwise standard ironing of the virtual term $\frac{\gamma\omega_b}{\lambda_T(\theta)}\haz(\theta)$ restores a nondecreasing $b^\ast(\cdot)$. \qedhere
\end{proof}

\begin{proof}[Proof of Proposition~\ref{prop:opt_cap}]
Work with the reduced form, threshold $\beta$, and IFR so that ironing yields a nondecreasing allocation. The Province minimizes
\[
\E_\theta\Big[\E\big[\alpha\,\min\{\beta(\hat G),b(\theta)\}+\tfrac{\kappa}{2}\min\{\beta(\hat G),b(\theta)\}^{2}\mid\theta\big]+\gamma\,T(\theta)\Big]
\]
subject to IC/IR/LL and monotonicity. Using Lemma~\ref{lem:capcalc}, the pointwise marginal cost (at type $\theta$) of increasing $b(\theta)$ equals $(\alpha+\kappa b(\theta))\,\PP_\theta[\beta\ge b(\theta)]$. By Myerson’s lemma, the virtual marginal benefit equals $(\gamma/\lambda_T)\,\omega_b\,\haz(\theta)\,\PP_\theta[\beta\ge b(\theta)]$ with $\haz(\theta)=f/\bar F$. Equating and canceling the common tail probability gives the interior solution
\[
b(\theta)=\frac{1}{\kappa}\left(\frac{\gamma\,\omega_b}{\lambda_T}\,\haz(\theta)-\alpha\right).
\]
Projection onto $[0,\bar b]$ yields \eqref{eq:b-star}, and IFR implies monotonicity (Lemma~\ref{lem:mono_cap}). To recover $T^\ast$, note that under $\lambda_T\equiv\omega_T$ the implemented index $x(\theta)=\omega_T\,T(\theta)+\omega_b\,\tilde b(\theta;\theta)$ must be nondecreasing; holding $x$ feasible implies \eqref{eq:T-star} a.e., with $T^\ast(\theta^{\min})=0$ by IR normalization. \qedhere
\end{proof}

\begin{proof}[Proof of Proposition~\ref{prop:nobailout}]
At $b\equiv 0$, the marginal expected cost of relaxing the cap at $\theta$ is $\alpha$ (Lemma~\ref{lem:capcalc}), while the virtual marginal benefit equals $(\gamma/\lambda_T)\,\omega_b\,\haz(\theta)=(\gamma\omega_b/\omega_T)\,\haz(\theta)$ under the threshold-$\beta$ baseline. If $\alpha\ge (\gamma\omega_b/\omega_T)\,\haz(\theta)$ for all $\theta$, then the KKT condition is nonnegative everywhere and $b=0$ is pointwise optimal. Otherwise at any $\theta^\sharp\in\arg\max \haz(\theta)$ with strict inequality, increasing $b(\theta^\sharp)$ strictly reduces the objective, so $b^\ast\equiv 0$ cannot be optimal. \qedhere
\end{proof}

\begin{proof}[Proof of Proposition~\ref{prop:welfare}]
Total expected welfare equals the sum of municipal interim utilities minus provincial costs:
\[
W(T,b)=\E_\theta[V(\theta)]-\E_\theta\Big[\gamma T(\theta)+\E\{\alpha p+\tfrac{\kappa}{2}p^2\mid\theta\}\Big],
\quad p=\min\{\beta(\hat G),b(\theta)\}.
\]
Under IC with $V(\underline\theta)=\underline U$, the envelope formula (Remark~\ref{rem:lambdaT-const}) gives
\[
V'(\theta)=\lambda_T'(\theta)T(\theta)+\omega_b\,\partial_\theta \tilde b(\theta;\theta)+K'(\theta).
\]
Integrating and substituting into $W$ yields a virtual-surplus functional in which the $\theta$–wise marginal effect of $b(\theta)$ is exactly the difference between the virtual benefit $(\gamma/\lambda_T)\omega_b \,\haz(\theta)$ and the marginal expected cost $\alpha+\kappa b(\theta)$ (times the common tail probability). Hence maximizing $W$ subject to IC/IR/LL and monotonicity is equivalent to the pointwise KKT condition used in Proposition~\ref{prop:opt_cap}; the resulting allocation is therefore second–best efficient among all IC–IR–LL mechanisms. \qedhere
\end{proof}
\clearpage   
\end{appendices}
\bibliographystyle{apalike}
\bibliography{references}

\end{document}